\documentclass[letterpaper, 10 pt, conference]{ieeeconf}  

\usepackage{graphicx}          
\let\labelindent\relax

\usepackage{grffile}
\usepackage{amsmath,amssymb,amsbsy}
\usepackage{hyperref}
\usepackage{amsthm} 
\usepackage{color}
\usepackage{comment}
\usepackage{lipsum}
\usepackage{graphicx}%
\usepackage{subfigure}
\usepackage[normalem]{ulem}
\usepackage{enumerate}
\usepackage{epstopdf}

\newtheorem{theorem}{Theorem}[section]
\newtheorem{lemma}[theorem]{Lemma}
\newtheorem{definition}[theorem]{Definition}	
\newtheorem{proposition}[theorem]{Proposition}
	
\newtheorem{remark}[theorem]{Remark}	
\newtheorem{problem}[theorem]{Problem}

\IEEEoverridecommandlockouts                              

\title{\LARGE \bf
Coverage and Field Estimation on Bounded Domains by \\ Diffusive Swarms
}

\author{Karthik Elamvazhuthi, Chase Adams, and Spring Berman
\thanks{*This work was supported by National Science Foundation (NSF)
award no. CMMI-1436960.}
\thanks{Karthik Elamvazhuthi, Chase Adams, and Spring Berman  are with the School for Engineering of
Matter, Transport and Energy, Arizona State University, Tempe, AZ, 85281
USA {\tt\small \{karthikevaz, Chase.Adams, Spring.Berman\}@asu.edu}}%
}

\begin{document}

\maketitle
\thispagestyle{empty}
\pagestyle{empty}

\begin{abstract}

In this paper, we consider stochastic coverage of bounded domains by a diffusing swarm of robots that take local measurements of an underlying scalar field. We introduce three control methodologies with diffusion, advection, and reaction as independent control inputs.  We analyze the diffusion-based control strategy using standard operator semigroup-theoretic arguments. We show that the diffusion coefficient can be chosen to be dependent only on the robots' local measurements to ensure that the swarm density converges to a function proportional to the scalar field. The boundedness of the domain precludes the need to impose assumptions on decaying properties of the scalar field at infinity. Moreover, exponential convergence of the swarm density to the equilibrium follows from properties of the spectrum of the semigroup generator. In addition, we use the proposed coverage method to construct a time-inhomogenous diffusion process and apply the observability of the heat equation to reconstruct the scalar field over the entire domain from observations of the robots' random motion over a small subset of the domain. We verify our results through simulations of the coverage scenario on a 2D domain and the field estimation scenario on a 1D domain. 

\end{abstract}


\section{INTRODUCTION}

Distributed control laws for multi-robot coverage strategies have been widely investigated \cite{cortes2002coverage,schwager2009decentralized}. Applications of coverage strategies include environmental monitoring, surveillance, source localization \cite{hayes2002distributed}, and vehicle scheduling \cite{pavone2011adaptive}.  In this work, we consider a variant of the coverage problem in which the goal is to achieve target coverage of an environment in a statistical sense. This is a significant departure from methods such as \cite{cortes2002coverage}, in which the robots are required to converge to a precise configuration in space and thus need more sophisticated sensing and control capabilities. Our approach can be applied to scenarios in which uncertainty in the strategy is beneficial, for instance in surveillance or source localization problems where the optimal coverage distribution is not known {\it a priori}.  It is also suitable for swarm robotic systems in which the severe resource constraints on the robots make it infeasible to implement global localization and extensive inter-robot communication.
 

Various stochastic methods for applications such as multi-robot task allocation and surveillance have been developed recently
\cite{acikmese2012markov,agharkar2014robotic,bandyopadhyay2013inhomogeneous,berman2009optimized,hamann2008framework,milutinovic2006modeling}. An important characteristic of many of these methods has been the index-free/permutation-invariant nature of the control laws, which can be beneficial for scalability in controller design \cite{leonard2013nonuniform,belta2004abstraction,kingston2011distributed}. This advantage of permutation invariance has led to multiple works on partial differential equation (PDE)-based multi-agent control, in which the Eulerian perspective of particles/agents is fundamental \cite{brockett2012notes,foderaro2014distributed,pimenta2013swarm}. 
The models that we present are largely based on those developed in our previous work \cite{berman2011design,elamvazhuthi2015optimal}, in which we used PDE optimization-based methods to synthesize robot controllers for stochastic coverage and task allocation problems in robotic swarms. In contrast to these works, our work here does not require knowledge of the target coverage distribution if this distribution depends on an environmental parameter, such as a scalar field that can be measured by the robots.

Our approach can be viewed as a variant of the method presented in \cite{mesquita2008optimotaxis} for unbounded domains. A similar problem was considered in \cite{huck2013stochastic} in the discrete-time case for agents with unicycle dynamics on bounded subsets of $\mathbb{R}^2$ for a source localization problem. The analysis in \cite{huck2013stochastic} proves the existence and uniqueness of some stationary distribution; however, it is unclear whether this distribution may be any desired invariant distribution. In our work, on the other hand, we introduce a family of control laws that can achieve any desired distribution that is uniformly bounded from above and below, thus enabling optimization over different coverage strategies, but we do not take robot kinematics into account. Given that controllable driftless systems can track a sufficiently rich set of trajectories arbitrarily well \cite{jean2014control}, we do not view this simplification as a significant disadvantage.  

In addition to the coverage strategy, we present a method for estimating the scalar field by observing random walks of robots over only a small subset of the domain. Our method exploits the (approximate) observability of the heat equation.  In this way, it is similar in approach to the work in \cite{devore2014recovery}, where the observability of the heat equation is used to recover the initial temperature of a rod from point measurements. Hence, our method relaxes the assumption, required by similar stochastic multi-agent approaches for estimating scalar fields \cite{mesquita2008optimotaxis,huck2013stochastic}, that agents are observed over the entire domain.  Our estimation method is suitable for independently operating, unidentified robots, unlike other multi-agent approaches to scalar field mapping, e.g. \cite{hayes2002distributed}, which rely on interactions between agents or require agents with unique identities.



Observability and controllability properties of the heat equation have been well-studied in the PDE control community \cite{curtain2012introduction,glowinski2008exact,zuazua2007controllability}. To characterize the uniqueness and stability of the desired invariant distribution of the stochastic process associated with the diffusion of the robots, we use the corresponding system of parabolic PDEs, which determine the evolution of transition probabilities over time. Toward this end, we consider the PDEs in an operator semigroup-theoretic framework, which enables a treatment of the PDEs as an abstract system of ordinary differential equations on an appropriately chosen function space.

Several complexities arise in the analysis of long-time behavior of linear semigroups on infinite-dimensional spaces, particularly the varied notion of the spectrum of a linear operator acting on these spaces. Even when the complete spectrum of the generator has been identified, the semigroup behavior might not be determined by spectral information alone. 


\section{PROBLEM FORMULATION} \label{sec:Problem}

We consider a swarm of agents that are deployed into a domain $\Omega$, a bounded convex open subset of $\mathbb{R}^n$ with Lipschitz continuous boundary $\partial \Omega$.  Each agent switches probabilistically between an {\it active} state, during which it explores the domain with a combination of deterministic and random motion, and a {\it passive} state, during which it stops to take a measurement.  The deterministic motion is governed by a time-dependent velocity $\mathbf{v}_1(t) \in \mathbb{R}^n$, and the random motion is represented as diffusion with an associated diffusion coefficient $v_2(t)$.  Diffusion can model probabilistic search, exploration, and tracking strategies or stochasticity arising from sensor and actuator noise.  An agent switches from the {\it active} state to the {\it passive} state at a time-dependent probability rate $v_3(t)$, and it switches back to the {\it active} state at a fixed probability rate $k$. The velocity $\mathbf{v}_1(t)$, diffusion coefficient $v_2(t)$, and state transition rate $v_3(t)$ are the control parameters of the system.  



Given these parameters, we can define a stochastic process  $(\mathbf{X}(t), Y(t))$, with state space $\Omega \times \lbrace 0,1 \rbrace$,  that models the motion of an agent with single-integrator dynamics and stochastic switching between states.  Here, $\mathbf{X}(t)$ is the position of the agent at time $t$, and $Y(t)$ is a switching variable that indicates whether the agent is in the {\it active} or {\it passive} state.  This variable is determined by the conditional probabilities $\mathbb{P}(Y(t+h) = 1 ~|~ Y(t) = 0) =  \int_h^{t+h} v_3(\tau) d \tau+ o(h^2) $, $\mathbb{P}(Y(t+h) = 0 ~|~ Y(t) = 1) = k h  + o(h^2)$.  $\mathbf{W}(t)$ is the standard Wiener process and $\psi(t)$ is the {\it reflecting function}, a process that characterizes the specular reflection of the agent at the boundary \cite{skorokhod1961stochastic,tanaka1979stochastic}. Then, the stochastic process $(\mathbf{X}(t), Y(t))$ satisfies a system of stochastic differential equations given by:  
\begin{align*}
&d\mathbf{X}(t) =  Y(t) \left(\mathbf{v}_1(t)dt +  \sqrt{2}v_2(t) d\mathbf{W} \right)  + d\mathbf{\psi}(t), \nonumber \\
&\mathbf{X}(0) = \mathbf{X}_0, \hspace{2mm}  Y(0) = Y_0.
\end{align*} 






We now present problems of coverage and estimation of an unknown scalar field $F : \Omega \rightarrow \mathbb{R}_+$ that is defined at each location $\mathbf{x} \in \Omega$.  We denote the normalized measure induced by the scalar field $F$ as $\mu_F$, where $\mu_F(d\mathbf{x}) = d\mathbf{x} \frac{F(\mathbf{x})}{\int_\Omega F(\mathbf{y})d\mathbf{y}}$. Here, $d\mathbf{x}$ is the Lesbesgue measure of the ``infinitesimal neighborhood" of $\mathbf{x}$. In addition, we define $\mu_{\mathbf{X}(t)}$ as the distribution associated with the random variable $\mathbf{X}(t)$ for each $t \geq 0$.  Our main objective is to design agent control laws that drive the swarm to a steady-state distribution that is proportional to the density of the field  $F(\mathbf{x})$, using only local measurements of the field.  We refer to this objective as a \textit{distributional controllability} problem and frame it as follows:  
\begin{problem} 
\label{prob:ADRctrl}
Determine whether there exist feedback control laws $D: \Omega \rightarrow \mathbb{R}$, $\mathbf{a} : \Omega \rightarrow \mathbb{R}^n$, and  $H: \Omega \rightarrow \mathbb{R}$ such that $\mu_{\mathbf{X}(t)}$ converges (weakly) to $\mu_F$ as $t \rightarrow \infty$ for the following stochastic process:
\begin{align}
\label{eq:chSDE} 
&d\mathbf{X}(t) =  Y(t)(\mathbf{a}(\mathbf{X}(t))dt +   \sqrt{2} D(\mathbf{X}(t)) d\mathbf{W})  + d\psi(t), \nonumber \\
&\mathbf{X}(0) = \mathbf{X}_0, \hspace{2mm}  Y(0) = Y_0, \hspace{4mm} t\geq0,
\end{align} 
where  $Y(t)$ is defined by the conditional probabilities $\mathbb{P}(Y(t+h) = 1 ~|~ Y(t) = 0) = \int_h^{t+h} H (X(\tau)) d \tau+ o(h^2) $, $\mathbb{P}(Y(t+h) = 0 ~|~ Y(t) = 1) = k h  + o(h^2)$.
\end{problem}
The spatiotemporal evolution of the population densities of agents that follow process \eqref{eq:chSDE} is described by a set of advection-diffusion-reaction PDEs.  We define $Q = \Omega \times (0,T) $ and $\Sigma = \partial\Omega \times (0,T) $ for some fixed final time $T$.   The vector $\mathbf{n}$ is the outward normal to $\partial\Omega$.  The densities of {\it active} and {\it passive} agents over the domain are denoted by $y_1(\mathbf{x},t)$ and $y_2(\mathbf{x},t)$, respectively.   Then the PDE model is given by 
\begin{eqnarray}
 \frac{\partial y_{1}}{\partial t} & = & \Delta(D(\mathbf{x})^2y_{1}) - \nabla \cdot (\mathbf{a(\mathbf{x})} y_{1}) \nonumber \\
 && - H(\mathbf{x})y_{1} + k y_{2} \hspace{2mm} ~~in \hspace{2mm} Q, \nonumber \\
  \hspace{0mm} \frac{\partial y_{2}}{\partial t} & = & H(\mathbf{x}) y_{1} - k y_{2}  \hspace{2mm} ~~in \hspace{2mm} Q,
\label{eq:Macro1}
\end{eqnarray} 
with the zero-flux boundary condition  
\begin{equation}
\mathbf{n} \cdot (\nabla (D(\mathbf{x})^2 y_1) - \mathbf{a(x)}y_1 )  = 0 \hspace{2mm} on \hspace{2mm} \Sigma
\label{eq:NFBC} 
\end{equation}
and initial conditions
\begin{equation}
y_1(\mathbf{x},0)= y_{10}(\mathbf{x}), \hspace{2mm} y_2(\mathbf{x},0) = y_{20} \hspace{2mm} on \hspace{2mm} \Omega.
\label{eq:MacIC}
\end{equation}




Using results from stochastic calculus \cite{pilipenko2014introduction},  we can determine that for the process satisfying \eqref{eq:chSDE}, we have $\mu_{\mathbf{X}(t)}(dx) =  d\mathbf{x}(y_1(\mathbf{x},t) +   y_2(\mathbf{x},t))$ for all $t \in [0,\infty)$. It can be shown that there are multiple sets of control laws that solve Problem \ref{prob:ADRctrl}. By setting the time derivatives of the PDEs in \eqref{eq:Macro1} equal to zero, for example, the choice $(D,\mathbf{a},H) = ( c_1/F^{1/2}+ c_2, ~c_2 \nabla F/F,~0)$ for any $c_1 >0$ and  $c_2 \geq 0$ makes the desired distribution invariant. 
 Additionally, $(D,\mathbf{a},H) = (c_1~,0~,c_2 F)$ approximates the target distribution within an arbitrary degree of accuracy for an appropriate choice of $c_1, c_2 > 0$. Note that each set of control laws only requires agents' local measurements of the scalar field $F(\mathbf{x})$. 
We leave the analysis of this general class of control laws to future work. 


For the remainder of this paper, we consider purely diffusion-based coverage, in which $\mathbf{a} = \mathbf{0}$ and $H \equiv 0$. Then the stochastic process \eqref{eq:chSDE} reduces to
\begin{equation}
\label{eq:Miceq}
d\mathbf{X}(t) =   \sqrt{2} D(\mathbf{X}(t)) d\mathbf{W}  + d \mathbf{\psi}(t).
\end{equation} 
The corresponding PDE model governs the density of {\it active} agents only, denoted here by $y(\mathbf{x},t)$:
\begin{eqnarray} 
\label{eq:Macdiff}
\frac{\partial y}{\partial t} & = & \Delta (D(\mathbf{x})^2 y) \hspace{2mm} ~~in \hspace{2mm} Q, \nonumber \\
\mathbf{n} \cdot \nabla (D(\mathbf{x})^2 y) & = & 0 \hspace{2mm} ~~on \hspace{2mm} \Sigma, \nonumber \\
y(\mathbf{x},0) & = & y_{0}(\mathbf{x})  \hspace{2mm} ~~in \hspace{2mm} \Omega.
\end{eqnarray} 

Given a swarm that performs diffusion-based coverage with an unknown control law $D(\mathbf{x})$, we want to additionally determine whether we can reconstruct this control law by observing the random motion of agents over a small subset of the domain. The estimation problem can be formulated as follows.  
Consider the SDE
\begin{equation}
\label{eq:ihetMiceq}
d\mathbf{X}(t) =   \sqrt{2} D(\mathbf{X}(t),t) d\mathbf{W}  + d \mathbf{\psi}(t), \hspace{2mm} t  \in [0,T_2],
\end{equation} 
where $D(\mathbf{x},t) = c/F(\mathbf{x})^{1/2}$ over time $t \in [0,T_1]$,  $T_1 \leq T_2$, for some $c>0$, and $D(\mathbf{x},t) = d>0$ otherwise.

\begin{problem}
\label{Pr2}
Let $O \subset \Omega$ be an open set, $G$ be a finite measurable partition of $O$, and $\lbrace \mathbf{X}^i(t) \rbrace$ be a set of $N$  i.i.d. random variables. Given $y_\omega(t) = \sum_{1 \leq i \leq N} \mathbf{1}_{\omega}(\mathbf{X}^i(t))/N$ for each $\omega \in G$, determine whether there exists a unique map $F: \Omega \rightarrow \mathbb{R_+}$ such that $\big \lbrace \mathbf{X}^i(t) \big \rbrace $ have the same distribution as the process satisfying \eqref{eq:ihetMiceq}.
\end{problem}

Note that in this problem, the identities of the agents are not important. Thus, similar to index-free control strategies, we can pose estimation problems in an index-free setting (e.g. \cite{elamvazhuthi2014variational,ramachandranoptimal,zeng2015ensemble}).  

\section{PRELIMINARIES} \label{sec:prelim}
In this section, we recall some standard notions from the theory of operator semigroups \cite{tucsnak2009observation}.
Let $H$ be a Hilbert space and $\mathcal{L}(H)$ be the space of bounded operators on $H$. 
\begin{definition}
A family $\mathbb{T} = ( \mathbb{T}(t) )_{t \geq 0}$ of operators in $\mathcal{L}(H)$ is a \textbf{strongly continuous semigroup} on $H$ if 
\begin{enumerate}
\item{$\mathbb{T}(0) = I$}
\item{$\mathbb{T}(t+\tau) = \mathbb{T}(t) \mathbb{T}(\tau)$}
\item{$lim_{t \rightarrow 0} \mathbb{T}(t)z =z \hspace{4mm} \forall z \in H $}
\end{enumerate}
\end{definition}
\begin{definition}
The linear operator $A : \mathcal{D}(A) \rightarrow H$, defined by
\begin{align*}
\mathcal{D}(A) &= \bigg \lbrace z \in H : \lim_{t \rightarrow 0, t >0} \frac{\mathbb{T}(t)z -z}{t}  \hspace{2mm} exists \bigg \rbrace, \\ 
Az &= \lim_{t \rightarrow 0, t >0} \frac{\mathbb{T}(t)z -z}{t} \hspace{4mm} \forall z \in \mathcal{D}(A),
\end{align*}
is called the \textbf{infinitesimal generator} (or just the \textbf{generator}) of the semigroup $\mathbb{T}$.
\end{definition}
\begin{remark}
Whenever we refer to the generator in this paper, we refer to the adjoint of the generator of the stochastic process. 
\end{remark}
\begin{definition}
An unbounded linear operator $A: \mathcal{D}(A) \rightarrow H$ is said to be \textbf{dissipative} if 
\begin{equation}
Re \big < Av, v \big > \leq 0 \hspace{4mm} \forall v \in \mathcal{D}(A)
\end{equation}
\end{definition}
Let $w \in L^{\infty}(\Omega)$ such that $\frac{1}{w} \in L^{\infty}(\Omega)$, i.e. $w$ is an essentially bounded Lesbesque measurable real-valued function with an essentially bounded inverse. Additionally, assume that $w$ is positive almost everywhere (a.e.) on $\Omega$. We define the space of square-integrable, real-valued measurable functions on $\Omega$, $L^2_w(\Omega)$, with the weighted 2-norm
\begin{equation}
\|f\|_{2w} = \bigg ( \int_{\Omega} |f(\mathbf{x})|^2 w(\mathbf{x})d\mathbf{x} \bigg )^{1/2},
\end{equation}
which is induced by the inner product $(\cdot,\cdot)_w:L^2_w(\Omega) \times L^2_w(\Omega) \rightarrow \mathbb{R}$, defined as
\begin{equation}
(f,g)_w = \int_{\Omega} f(\mathbf{x})g(\mathbf{x}) w(\mathbf{x})d\mathbf{x}.
\end{equation}
By Holder's inequality, we have that $c_1 \|f\|_{2w}  \leq \|f\|_{2} \leq c_2 \|f\|_{2w} $ for some $c_1 >0$ and $c_2 >0$. Hence, $L^2(\Omega) \simeq L^2_w(\Omega)$, i.e. the spaces are isomorphic.
 
In the forthcoming definitions, all derivatives with respect to spatial variables are to be understood as weak/distributional derivatives.
Define $H^1_w(\Omega)$ as
\begin{equation}
\begin{split}
H^1_w(\Omega) = \bigg \lbrace f \in L^2_w(\Omega) : \frac{\partial}{\partial x_{\alpha} } (w f) & \in L^2_w(\Omega) \\
 \forall \alpha \in \lbrace 1,2,....n \rbrace \bigg \rbrace. &
\end{split}
\end{equation}
The norm on this space is induced by the inner product 
\begin{equation} \label{IP1}
\begin{split}
(f,g)_{H^1_w(\Omega)} & = \int_{\Omega} f(\mathbf{x})g(\mathbf{x}) w(\mathbf{x})d\mathbf{x}  \\
 & \hspace{2mm} +~ \sum_{i =1}^{n} \int_{\Omega} \frac{\partial (w f)}{\partial x_i } (\mathbf{x}) \frac{\partial (w g)}{\partial x_i }(\mathbf{x}) d\mathbf{x}.
\end{split}
\end{equation}
Additionally, we define the space 
\begin{equation}
\begin{split}
H^2_w(\Omega) = \bigg \lbrace f \in H^1_w(\Omega) : \frac{\partial^2 (w f)}{\partial x^2_{\alpha}}  & \in L^2_w(\Omega) \\
 \forall \alpha \in \lbrace 1,2,....n \rbrace \bigg \rbrace, &
\end{split}
\end{equation}
which is equipped with the inner product 
\begin{equation} \label{IP2}
\begin{split}
(f,g)_{H^2_w(\Omega)} & = \int_{\Omega} f(\mathbf{x})g(\mathbf{x}) w(\mathbf{x})d\mathbf{x}  \\
& \hspace{2mm} + \sum_{j =1}^{2} \sum_{i =1}^{n} \int_{\Omega} \frac{\partial^j (w f)}{\partial x^j_i } (\mathbf{x}) \frac{\partial^j (w g)}{\partial x^2_i }(\mathbf{x}) d\mathbf{x}
\end{split}
\end{equation}
Note that for $w \equiv 1$, $H^1_w(\Omega)$ and $H^2_w(\Omega)$ are the same as the traditional Sobolev spaces $H^1(\Omega)$ and $H^2(\Omega)$, respectively. 
We can then consider the PDE \eqref{eq:Macdiff} as an abstract system of ordinary differential equations on $L^2_w(\Omega)$,
\begin{eqnarray}
\label{eq:AbsMacro}
\dot{\mathbf{y}}(t) &=& A \mathbf{y}(t) \hspace{4mm} (t \geq 0) \\  \nonumber
\mathbf{y}(0) &=& \mathbf{y}_0
\end{eqnarray} 
\begin{equation} 
\label{eq:PDGen}
Af = \Delta (wf),
\end{equation}
with $\mathcal{D}(A) = \bigg \lbrace f \in H^2_w : \mathbf{n} \cdot \nabla(wf(\mathbf{x})) = 0 \hspace{2mm} \forall \mathbf{x} \in \partial \Omega  \bigg \rbrace $ and the corresponding norm. The requirements on the behavior of functions in $\mathcal{D}(A)$ on the boundary of $\Omega$ are to be understood in the ``trace sense.'' Since $\mathcal{D}(A)$  is a subset of $H^2_w(\Omega)$ and $\partial \Omega$ is at least Lipschitz, the trace operation corresponding to the normal derivative is well-defined. From here on, we focus our analysis on the system \eqref{eq:Macdiff}. We establish that the operator defined in \eqref{eq:PDGen} generates a semigroup on $L^2_w(\Omega)$.

The main advantage of considering the weighted space $L^2_w(\Omega)$, rather than $L^2(\Omega)$, is that the operator $A$ as defined in \eqref{eq:PDGen} is self-adjoint as an operator on the former function space, and hence this simplifies much of the analysis.

\begin{remark}
By working in the $L^2$ framework, we are tacitly assuming that the distribution function of the initial condition of the stochastic process is square-integrable. Due to the inclusion of $L^2(\Omega)$ in $L^1(\Omega)$, whenever $\Omega$ is a bounded domain, this assumption is not too restrictive. Moreover, given the ultracontractivity of the semigroup of interest (not proved in this work), initial conditions in $L^1(\Omega)$ are mapped to $L^{\infty}(\Omega) \subset L^2(\Omega) \subset L^1(\Omega)$ for any $t >0$.
\end{remark}

\vspace{0mm}

\section{ANALYSIS} \label{sec:analysis}


\subsection{Coverage} \label{sec:Coverage}

In this section, we derive a result (Theorem \ref{covTheorem}) that the choice of the control law $D(\mathbf{X}) = c/F(\mathbf{X})^{1/2}$ in process \eqref{eq:Miceq}, where $c > 0$, yields $lim_{t \rightarrow \infty}\mu_{\mathbf{X}(t)}  \rightarrow \mu_F$ as specified in Problem \ref{prob:ADRctrl}.  In fact, we establish a stronger form of convergence than the convergence required by Problem \ref{prob:ADRctrl}. Hence, the agent control law that solves the coverage problem for a purely diffusive swarm is dependent only on pointwise observations of the scalar field $F$. 

We first introduce several results that are needed to prove Theorem \ref{covTheorem}.  Here, the operator $A$ has the definition in equation \eqref{eq:PDGen}.

First, we establish that $A$ generates a semigroup, and hence a unique mild solution of the PDE \eqref{eq:Macdiff} exists.

\begin{proposition}
$A$ is a dissipative operator and generates a strongly continuous semigroup on $L^2_w(\Omega)$.
\end{proposition}
\begin{proof}
Using integration by parts, it can be verified that for each $z \in L^2_w(\Omega)$, $\big<Az,z \big>_{L^2_w(\Omega)} \leq 0$. Hence, $A$ is a dissipative operator.
To show that $A$ generates a strongly continuous semigroup on $L^2_w(\Omega)$, we first define the bilinear form $B : H^1_w \times H^1_w \rightarrow \mathbb{R} $ by
\begin{equation}
\begin{split}
B(u,v) & = \int_\Omega \nabla (w(\mathbf{x})u(\mathbf{x})) \cdot \nabla (w(\mathbf{x})v(\mathbf{x})) d\mathbf{x} \\
\end{split}
\end{equation}
Then we have that
\begin{equation}
|B(u,v)| \leq \|u\|_{H^1_w(\Omega)} \|v\|_{H^1_w(\Omega)}
\end{equation}
for all $u,v \in H^1_w(\Omega)$.
In addition, for  all $u \in H^1_w(\Omega)$ such that $\int_\Omega w(\mathbf{x}) u(\mathbf{x})d\mathbf{x} = 0$, the following inequality holds for some $c>0$:
\begin{equation}
|B(u,u)| > c \|u\|^2_{H^1_w(\Omega)}.
\end{equation}

Then, by the Lax-Milgram theorem \cite{brezis2010functional}[Corollary 5.8], we can state that for each $f \in L^2_w(\Omega)$ such that $\int_\Omega w(\mathbf{x})f(\mathbf{x})d\mathbf{x} = 0$, there exists a unique solution $u \in H^1_w(\Omega)$ to the equation
\begin{equation}
B(u,v)  =  (f,v)_{L^2_w(\Omega)}
\end{equation}
for all $v \in H^1_w(\Omega)$. By a similar argument, a solution $u \in H^1_w(\Omega)$ exists for each $f \in L^2_w(\Omega)$.

Moreover, each solution $u$ is in fact in $\mathcal{D}(A)$. Hence, we have that $R(I - A) = L^2_w(\Omega)$, where $R(\cdot)$ denotes the range of the operator.  Therefore, the result follows from the dissipativeness of $A$ and \cite{engel2000one}[Chapter II, Corollary 3.2].
\end{proof}


\begin{proposition}
\label{Cresol}
The unbounded operator $A$ in \eqref{eq:PDGen} has a compact resolvent.
\end{proposition}
\begin{proof}
Consider the Neumann Laplacian $\Delta_N$, defined on $L^2(\Omega)$, with domain $\mathcal{D}(\Delta_N) = \big \lbrace f \in H^2 : \mathbf{n} \cdot \nabla(f(\mathbf{x})) = 0$ $\forall \mathbf{x} \in \partial \Omega  \big \rbrace $. The embedding $\gamma:\mathcal{D}(A) \rightarrow L^2_w(\Omega)$ can be written as a composition of maps $i \circ j \circ k$. Here, $i$ is the natural isomorphism $f \mapsto w \cdot f$ from $\mathcal{D}(A)$ to $\mathcal{D}(\Delta_N)$, $j$ is the embedding of $\mathcal{D}(\Delta_N)$ into $L^2(\Omega)$, and $k$ is the isomorphism $f \mapsto f/w$ from $L^2(\Omega)$ to $L^2_w(\Omega)$. Since $i$ is a compact operator,  for any bounded sequence $(u_n) \in \mathcal{D}(A)$ there exists a subsequence $(u_m)$ such that $\gamma(u_m) = (i \circ j \circ k)(u_m)$ is a convergent sequence in $L^2_w(\Omega)$. Hence, the embedding $\mathcal{D}(A) \hookrightarrow L^2_w(\Omega)$ is compact.
\end{proof}

\begin{lemma}
$A$ generates an immediately compact semigroup.
\label{ECLemm}
\end{lemma}
\begin{proof}
Let $(\mathbb{T}(t))_{t>0}$ be the strongly continuous semigroup generated by $A$. First, we note that $(\mathbb{T}(t))_{t>0}$ is analytic since $A$ is self-adjoint and negative. Hence, $(\mathbb{T}(t))_{t>0}$ is immediately norm continuous \cite{engel2000one}. Moreover, we note that $A$ has a compact resolvent from Proposition \ref{Cresol}. The result then follows from \cite{engel2000one}[Chapter II, Theorem 4.29]. 
\end{proof}

Finally, using the results above, we can demonstrate that the unique equilibrium of \eqref{eq:AbsMacro} is exponentially stable.


\begin{theorem} \label{covTheorem}
For any $y_0 \in L^2_w(\Omega)$ such that $y_0 \geq 0$ (i.e., $y_0$ is positive a.e. on $\Omega$), the semigroup $(\mathbb{T}(t))_{t\geq 0}$ generated by $A$ satisfies 
\begin{equation}
\|\mathbb{T}(t)y_0 - c/w\|_{2w} ~\leq~ M e^{-\omega t} \| y_0 - c/w\|_{2w}
\label{eq:stab}
\end{equation}
for all $t\geq 0$, some $\omega, c>0$, and $M \geq 1$.
\end{theorem}
\begin{proof}
First, we show that the integral of the solution $y(\mathbf{x},t)$ over the domain $\Omega$ remains conserved for any initial condition $y_0 \in L^2_w(\Omega)$. Let $u \in \mathcal{D}(A)$. We define a linear map $ R_t : \mathcal{D}(A) \rightarrow \mathbb{R}$ by
\begin{equation}
R_t u = \int_{\Omega} (\mathbb{T}(t)u - u) d\mathbf{x}
\end{equation}
for some $t>0$. Then, using Green's formula for twice weakly-differentiable functions, we have
\begin{align*}
R_t u &= \int_{\Omega} \int^t_0A\mathbb{T}(s)u_0dsd\mathbf{x} \nonumber \\ 
&= \int^t_0 \int_{\partial \Omega} \mathbf{n} \cdot \nabla(w(\mathbf{y})u(\mathbf{y},s)) d\mathbf{y}ds = 0
\end{align*}
from the boundary condition that is encoded in the definition of $A$ in \eqref{eq:PDGen}. Due to the boundedness of the map $R_t$ by \cite{pedersen2012analysis}[Proposition 2.1.11], this map can be extended to a bounded linear operator from $L^2_w(\Omega)$ to $\mathbb{R}$ (since $\mathcal{D}(A)$ is dense in $L^2_w(\Omega)$). Thus, $R_t$ is in fact the zero map for every $t \geq 0$. Hence, the integral of the solution $\mathbb{T}(t)y_0$ over the domain remains conserved. 

To prove the uniqueness and stability of the equilibrium, we make the following observations. It is well-known that the Neumann Laplacian has a unique one-dimensional linear subspace of constant eigenvectors (constant functions) corresponding to the eigenvalue $0$ \cite{evans1998partial}. Since there is a natural bijective correspondence between eigenvectors of $\Delta_N$ and $A$ for the eigenvalue $0$, this implies that $A$ has a unique one-dimensional subspace of eigenvectors, spanned by the function $1/w$, corresponding to the eigenvalue $0$. Therefore, $0$ is a first-order pole of $A$. The semigroup under consideration is eventually compact by Lemma \ref{ECLemm}, since immediate compactness implies eventual compactness. Additionally, we can choose $c = \frac{\int_\Omega  u d\mathbf{x}} {\int_\Omega 1/w d\mathbf{y}}$ in condition \eqref{eq:stab} since the integral of the solution over the domain must be conserved. Then from the positivity of the operator $-A$, whose spectrum therefore lies in the closed left-half plane, the result follows from the above arguments and \cite{engel2000one}[Chapter V, Corollary 3.3].
\end{proof}
 
As an alternative to the functional analytic methods used here, it is possible to use probabilistic approaches to establish asymptotic stability of the desired distribution. See for example \cite{meyn2012markov}, where such problems have been addressed for discrete-time Markov processes. Similar methods exist in the literature for Markov processes that evolve in continuous time.  Moreover, one can consider several other notions of stability, such as stability in the total variation norm, the Wasserstein distance, and convergence of Ces\`{a}ro means.
 
An additional issue is the well-posedness of process \eqref{eq:Miceq} when the control law $D(\mathbf{x}) = c/F(\mathbf{x})^{1/2}$ is implemented. A sufficient condition for the well-posedness of an ODE or SDE is that the coefficients are locally Lipschitz everywhere.  We note that global Lipschitzness of $c/F(\mathbf{x})^{1/2}$ is ensured whenever $F$ is globally Lipschitz on $\Omega$, positive, and uniformly bounded from below away from zero. Hence, there exists a sufficiently rich class of scalar fields $F$ that can be used to define the control law $D(\mathbf{x})$.

\begin{remark} 
Analogously, we can consider a similar diffusion process on a graph that is closely related to the Metropolis-Hastings algorithm. If $\mathcal{G}=(\mathcal{V},\mathcal{E})$ is a connected graph with undirected edges and $f:\mathcal{V} \rightarrow \mathbb{R}_+$ is a scalar field on the graph, then we can consider the continuous-time Markov chain on the graph whose generator is defined as $-\mathbf{D}\mathbf{L}$, where $\mathbf{L}$ is the Laplacian of the graph and $\mathbf{D}$ is a diagonal matrix with entries $D_{ii} = cf(i)$ for each $i \in \mathcal{G}$ and a fixed $c>0$. Then the evolution of transition probabilities $\mathbf{p} \in \mathbb{R}^{|\mathcal{V}|}$ is given by
\begin{equation}
\dot{\mathbf{p}}(t) = -\mathbf{LD}\mathbf{p}(t)~, \hspace{0.5cm} \mathbf{p}(0) = \mathbf{p}_0~. \nonumber
\end{equation}
We can view the above equation as the discretized approximation of the PDE \eqref{eq:Macdiff} for the case where $\mathcal{G}$ is a lattice graph. Note that the transition probabilities depend only on local information, as in the case of the diffusion process \eqref{eq:Miceq}. It is straightforward to check the invariance of the distribution $\mathbf{\pi} = \frac{\sum_{i \in \mathcal{V}}f(i)}{f}$.
\label{Rem1}
\end{remark}


\subsection{Field Estimation}
\label{FEstAn}


Our method for estimating a scalar field $F:\Omega \rightarrow \mathbb{R}_+$ from observations of agents consists of three steps, which we describe and justify in this section:
\begin{enumerate}
\item{\textit{Convergence}: We assume that all agents know the time parameter $T_1$.  During the time interval $t \in [0, T_1]$, agents  follow the closed-loop coverage control law $\mathbf{X}(t) =   (2/F(\mathbf{X}(t))^{1/2} d\mathbf{W}  + d \mathbf{\psi}(t)$. By the analysis in Section \ref{sec:Coverage}, the agents will converge to the distribution corresponding to $F$.} 
\item{\textit{Dispersion}: During the time interval  $t \in (T_1, T_2]$, the agents perform a homogenous random walk, their positions evolving according to $\mathbf{X}(t) =   \sqrt{(2 c)} d\mathbf{W}  + d \mathbf{\psi}(t)$ for some known $c>0$.}
\item{\textit{Estimation}: During the same time interval $t \in (T_1, T_2]$, an observer collects data $\lbrace y_{\omega} \rbrace$ over some finite partition of the domain of observation $O$, as described in Problem \ref{Pr2}, and solves the optimization problem that we present in this section (Theorem \ref{eq:MainOpt}).}
\end{enumerate}

\vspace{2mm}
We use the following result to justify our method.
\begin{theorem}
Let $\Omega$ be a bounded subset of $\mathbb{R}^n$ with a $C^2$ boundary, $O \subset \Omega$ be an open subset, $H_f$ be a finite-dimensional subspace of $L^2(\Omega)$, $\hat{y} \in L^2(T_1,T_2; L^2(O))$, and $d >0$. Then the following problem is well-posed and has a unique solution: 
\begin{align*}
\min_{u_T \in H_f}  & \|y - \hat{y}\|^2_{L^2(T_1,T_2; L^2(O))}  \\
s.t. \hspace{5mm}  \frac{\partial u}{\partial t} & = d \Delta u \hspace{2mm} in \hspace{2mm} Q, \\
\mathbf{n} \cdot \nabla u & = 0 \hspace{4mm} on \hspace{2mm} \Sigma,  \nonumber \\
\hspace{4mm}      y(t) &= C u(t) \hspace{2mm} \forall  t \in (T_1,T_2] ,\\
\hspace{4mm}      u(T_1) &= u_T,
\end{align*}
\label{eq:MainOpt}
\noindent where $C: L^2(T_1,T_2:L^2(\Omega)) \rightarrow  L^2(T_1,T_2:L^2(O))$ is the observation operator defined as  $y(\mathbf{x},t)= Cu(\mathbf{x},t) = u(\mathbf{x},t)$ for each $(\mathbf{x},t) \in O \times (T_1,T_2]$.
\end{theorem}
\noindent This result follows from the approximate observability of the heat equation with Neumann boundary condition \cite{fernandez2006global}. 
For the possibility of extending this result to more general domains, see \cite{apraiz2012observability}.
We note that we have only approximate observability of the heat equation. While $H_f$ can be any finite-dimensional subspace of $L^2(\Omega)$, we cannot replace $H_f$ by $L^2(\Omega)$ (or any infinite-dimensional subspace of $L^2(\Omega)$) and retain a well-posed problem with a unique solution. This requires having exact observability, which is generally only true in the trivial case where $O = \Omega$, that is, the evolution of the process can be observed over the entire domain. 
\begin{remark}
Here and in the following arguments, by treating Problem \ref{Pr2} as a PDE-constrained optimization problem, we are implicitly considering an idealized version of this problem in which $N \rightarrow \infty$ and $\omega$ is taken over all measurable subsets of $O$.
\end{remark} 

\vspace{2mm}

Then we can consider the PDE model
\begin{align*}
 \frac{\partial u}{\partial t} & =  \Delta( D(\mathbf{x},t)^2 u) \hspace{2mm} in \hspace{2mm} Q, \\
\mathbf{n} \cdot \nabla u & = 0 \hspace{4mm} on \hspace{2mm} \Sigma,  \nonumber \\
\hspace{4mm}      y(t) &= C u(t) \hspace{2mm} \forall  t \in (T_1,T_2], \\
\hspace{4mm}     u(\mathbf{x},0) & = u_{0}(\mathbf{x})  \hspace{2mm} in \hspace{2mm} \Omega,  
\end{align*}
where  $D(\mathbf{x},t) = c/F(\mathbf{x})^{1/2}$ over time $t \in [0,T_1]$ for some $c>0$, and $D(\mathbf{x},t) = d>0$ otherwise. From the analysis in Section \ref{sec:Coverage}, we know that $\|u(\mathbf{x},T_1) - F(\mathbf{x})/ \int_{\Omega}F(\mathbf{x}) d\mathbf{x} \| \leq \epsilon(T_1)$, where $\epsilon(T_1)$ is the error between the distribution at time $T_1$ and the desired distribution, for all initial conditions $u_0$ of this PDE model such that $\int_\Omega u_0 d\mathbf{x} =1$ and $u_0 \geq 0$. Moreover, $\epsilon(T_1) \rightarrow 0$ as $T_1 \rightarrow \infty$ from the stability estimate \eqref{eq:stab}. Therefore, by observing the random walks of agents from time $T_1$ to $T_2$, the observer can infer the density of agents at time $T_1$, and hence obtain an approximate estimate of the field $F(\mathbf{x})$ over $\Omega$ up to a proportional constant. Additionally, if the observer has an estimate of $F(\mathbf{x})$ for all $\mathbf{x} \in O$, then the proportional constant can be computed as well. 

\begin{remark}
This technique has a graph analogue, as in Remark \ref{Rem1}. Observability of consensus protocols on communication networks has been well-studied  \cite{mesbahi2010graph} and applied to problems of sensing spatially-distributed parameters  \cite{ ji2007observability,pequito2013optimal,ramachandran2016effect}. In these works, agents are communication nodes of the graph, whereas in our approach, they would be viewed as random-walking agents on the graph.
\end{remark} 

\begin{remark}
The agents will attain the steady-state distribution $\mu_F$ only in infinite time, but the time $T_1$, which is defined {\it a priori}, is necessarily finite.  The agent distribution will converge toward $\mu_F$ at an exponential rate that depends on the underlying scalar field.  Since we do not assume prior knowledge about this field, it is not possible to predict the degree of convergence at time $T_1$, which will affect the error in the subsequent estimate of the field.  Hence, an inaccurate estimate may result if $T_1$ is set to be too small for the agent distribution to have converged closely to $\mu_F$ at that time.
\end{remark}


\vspace{0.5mm}

\section{SIMULATIONS}
\label{sec:Sim}

We validated our coverage approach in two different simulated scenarios. In case $1$, the scalar field is defined as $F_1(\mathbf{x}) = f_1(\mathbf{x})-f_2(\mathbf{x}) + \epsilon$ for all $\mathbf{x} \in \Omega$, where $f_n$, $n=1, 2,$  are given by 
\begin{align*}
f_n(\mathbf{x}) &= \exp \bigg( \frac{-1}{1- \|a_n\mathbf{x} - b_n \|^2} \bigg)  &\mbox{for} \hspace{2mm} \|a_n\mathbf{x} - b_n \|^2 <1, \nonumber \\ 
&= 0 ~~~ \mbox{otherwise.} &
\end{align*}
We set $a_1 = 2, ~a_2 = 6, ~b_1 = 1, ~b_2 = 2,$ and $\epsilon = 0.01$.  The field $F_1(\mathbf{x})$ is shown in the lower right plot of Fig. \ref{fig:cas1}.
In case 2, we used the numerically constructed scalar field $F_2(\mathbf{x})$ that is shown in the lower right plot of Fig. \ref{fig:cas2}.


\begin{figure}[t]
\centering
\includegraphics[trim= 7mm 0mm 0mm 0mm, scale=0.7]{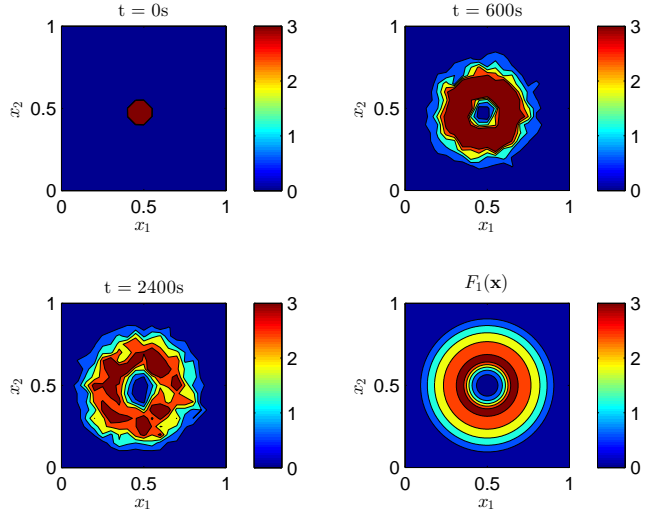}
\caption{Case 1: Simulated agent densities at three times $t$ and the underlying scalar field}
\label{fig:cas1}
\end{figure}
\begin{figure}[t]
\centering
\includegraphics[trim= 7mm 0mm 0mm 0mm, scale=0.7]{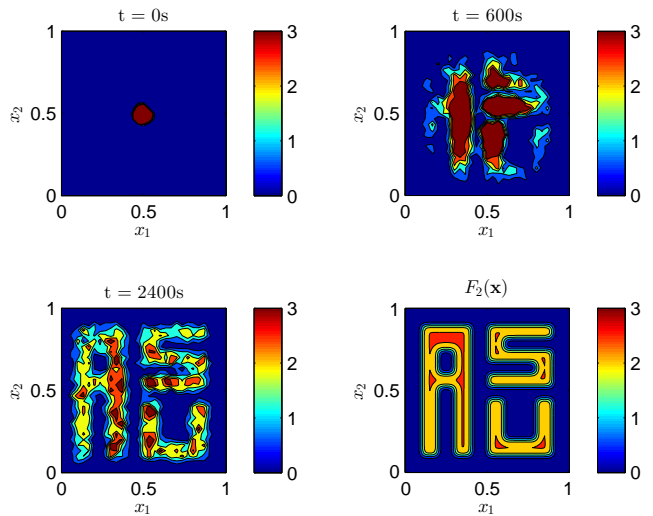}
\caption{Case 2: Simulated agent densities at three times $t$ and the underlying scalar field}
\label{fig:cas2}
\end{figure}

In each case, the diffusion-based feedback control law was chosen to be $D_n(\mathbf{x}) = 10^{-5} / F_n(\mathbf{x})^{1/2}$, $n=1,2$. Since $D_n$ is in $C^{\infty}(\bar{\Omega})$ and is uniformly bounded from below away from zero, it is globally Lipschitz on $\Omega$. For each case, $3000$ agents were simulated on a domain $\Omega = (0,1) \times (0,1)$. The agents were initially distributed as a Gaussian centered at $(0.5,0.5)$.
The stochastic motion of each agent was approximated in discrete time using the standard-form Langevin equation:  
\begin{equation}
\mathbf{X}(t+\Delta t) - \mathbf{X}(t) = (2 D_n^2(\mathbf{X}) \Delta t)^{1/2}~\mathbf{Z}(t), \label{eq:RWPTpos}
\end{equation}
where $\mathbf{Z} \in \mathbb{R}^2$ is a vector of independent, standard normal random variables.  When an agent encounters the boundary, it performs a specular reflection. As shown in Fig. \ref{fig:cas1} and \ref{fig:cas2}, the steady-state swarm density closely matches  the underlying scalar field in each case.

The field estimation algorithm was validated for two example scalar fields on a $1D$ domain, $\Omega = [0,1]$. These fields were defined as $F_1(x) = c_1 (\sin(\pi x) + 0.01)$ and $F_2(x) = c_2 (x^2 + 0.01)$ for all $x \in \Omega$, where $c_1$ and $c_2$ are normalization constants chosen such that the field integrates to $1$ over the domain. The region of measurement was set to $O = (0.7,1)$ in each case. The agent motion was simulated using the numerical approximation \eqref{eq:RWPTpos}. The measurement data $\lbrace y_{\omega} \rbrace$ was collected from both a coarse partition ($\cup_{n \in \mathbb{Z}_+}[\frac{n-1}{10},\frac{n}{10}) \cap O$) and a finer one ($\cup_{n \in \mathbb{Z}_+}[\frac{n-1}{100},\frac{n}{100}) \cap O$). The optimization problem in Theorem \ref{eq:MainOpt} was solved using an ``Optimize-then-Discretize'' approach \cite{pinnau2008optimization}, along with a projected gradient descent method. The objective functional was modified to its regularized version, $\|y-\hat{y}\|^2_2 +  \lambda \|u_T\|^2_2$, where $\lambda$ was chosen to be $0.1$. Since the problem in Theorem \ref{eq:MainOpt} is a convex optimization problem with linear (albeit infinite-dimensional) constraints, it is fairly straightforward to construct the optimality system, which consists of necessary and sufficient conditions associated with the adjoint equation that the optimal solution must satisfy. We exclude the analytical formulation of the gradient here for the sake of brevity. 

The results of the estimation procedure are illustrated in Fig. \ref{fig:sin} and \ref{fig:x2}. The observation data from the fine grid can be seen to yield a more accurate reconstruction of the scalar field than the data from the coarse grid. The estimation procedure performs the best with larger numbers of agents, as would be expected due to the relatively smaller amount of noise in the  data from larger populations. However, it is notable that the method works well, qualitatively at least, with populations of only 100 agents, which yield observation data with very large fluctuations from the mean behavior.

\begin{figure}[t]
\centering
\includegraphics[scale=0.6]{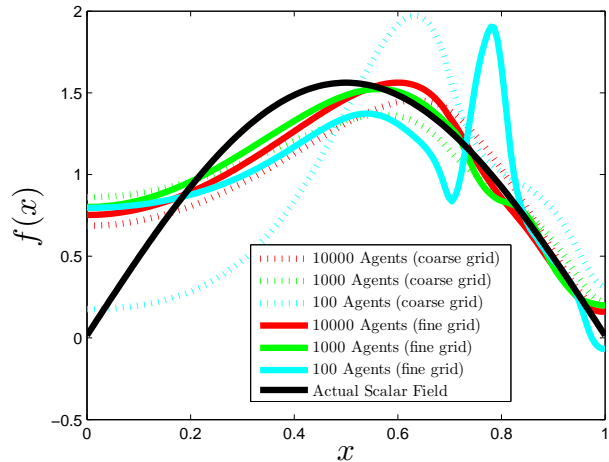}
\caption{$F_1(x)=c_1(\sin(\pi x)+0.01)$}
\label{fig:sin}
\end{figure}
\begin{figure}[t]
\centering
\includegraphics[scale=0.6]{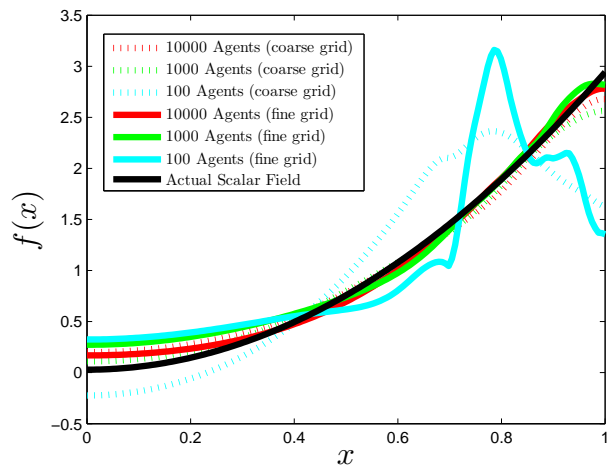}
\caption{$F_2(x)=c_2(x^2+0.01)$}
\label{fig:x2}
\end{figure}

\vspace{2mm}

\section{CONCLUSIONS}
In this work, we have developed a diffusion-based approach to achieving a spatial distribution of swarm activity that matches an underlying scalar field in the case where the agents have only local sensing, heading information, and no global position information or communication. We also presented a method for mapping scalar fields using observations of agents' random walks over a small subset of the domain by exploiting the observability properties of the heat equation and its relation to random walks.

In future work, we will analyze the advection- and reaction-based coverage schemes presented in this paper and compare the relative advantages of each strategy. It would also be useful to numerically compute the rate of convergence of the swarm to a desired distribution using techniques such as spectral approximations \cite{chatelin1983spectral} and sum of squares methods \cite{papachristodoulou2006analysis,meyer2015stability} for systems with polynomial data. Additionally, we will develop coverage strategies for agents with more complex dynamics and cooperative behaviors.  In future work on our field estimation approach, we plan to improve the efficiency of the numerical method for solving the optimization problem. Investigating the numerical well-posedness of controllability and observability problems of the heat equation is a challenge in itself \cite{munch2010numerical}. We will also consider Problem \ref{Pr2} in a more natural setting as a PDE coefficient identification problem \cite{isakov2006inverse}, for which specification of the $T_1$ time parameter would not be required.

\bibliographystyle{plain}

\bibliography{cdcref}

\end{document}